\documentclass{article}
\usepackage{spconf,preamble}
\usepackage{graphicx,hyperref}

\usepackage{xcolor} % text color 
\usepackage{enumitem} % label for {enumerate}
\usepackage{amsthm} % theorem environment

\newtheorem{theorem}{Theorem} 

\usepackage{algorithm} % for writing algorithm
\usepackage{algorithmic} % for writing algorithmic environment inside algorithm env

\newcommand{\algorithmicparameters}{\textbf{Parameters:}} % define \Parameters command
\newcommand{\PARAMETER}{\item[\algorithmicparameters]}
\newcommand{\algorithmicpreparation}{\textbf{Preparation:}} % define \PREPARATION command
\newcommand{\PREPARATION}{\item[\algorithmicpreparation]}

\title{Bayesian Signal Separation via Plug-and-Play\\ Diffusion-Within-Gibbs Sampling}
\name{Yi Zhang\thanks{This research was supported by the European Research Council (ERC) under the European Union’s Horizon 2020 research and innovation program (grant No. 101000967) and by the Israel Science Foundation (grant No. 536/22). This work is also supported by Manya Igel Centre for Biomedical Engineering and Signal Processing.}, Rui Guo\thanks{Corresponding author: Rui Guo (rui.guo@weizmann.ac.il).}, Yonina C. Eldar
}
\address{Faculty of Math and Computer Science, Weizmann Institute of Science, Rehovot, Israel}

\begin{document}
%\ninept
%
\maketitle
\begin{abstract}
We propose a posterior sampling algorithm for the problem of estimating multiple independent source signals from their noisy superposition. The proposed algorithm is a combination of Gibbs sampling method and plug-and-play (PnP) diffusion priors. Unlike most existing diffusion-model-based approaches for signal separation, our method allows source priors to be learned separately and flexibly combined without retraining. Moreover, under the assumption of perfect diffusion model training, the proposed method provably produces samples from the posterior distribution. Experiments on the task of heartbeat extraction from mixtures with synthetic motion artifacts demonstrate the superior performance of our method over existing approaches.  
\end{abstract}
\begin{keywords}
signal separation, source separation, diffusion, Gibbs sampling, plug-and-play
\end{keywords}

\section{Introduction}
\label{sec:intro}

We consider the problem of recovering $K$ independent source signals $\{\sv_k\}_{k=1}^K$ from their noisy superposition:
\begin{equation}\label{eq:observation}
	\yv = \sum_{k=1}^K \sv_k + \vv,
\end{equation}
where the observational noise $\vv \sim \Normal{\zerov}{\sigma_v^2 \Imat}$ is independent of sources. 
This formulation arises in diverse applications, such as speech separation in the ``cocktail party problem" \cite{luo2018}, clutter suppression in ultrasound imaging \cite{wildeboer2020}, and vital sign monitoring using radar \cite{han2023cough}. 

Classical model-driven approaches often rely on simplified priors on the sources, such as periodicity or narrowband structure, leading to decomposition methods like empirical mode decomposition (EMD \cite{huang1998}) or variational mode decomposition (VMD \cite{dragomiretskiy2014}). Blind separation techniques \cite{comon2010}, such as independent component analysis (ICA) and nonnegative matrix factorization (NMF), exploit statistical independence of the sources or nonnegativity assumptions, but are mainly tailored for multi-channel observations.  

With the advent of deep learning, data-driven approaches have become dominant. Earlier methods typically adopt a discriminative paradigm \cite{weninger2014, luo2018, wang2023}, training a neural network to map noisy mixtures to separated sources. While effective, such approaches usually require retraining when the number or prior distributions of sources change, and they only yield point estimates, offering no uncertainty quantification.   

Recently, generative models have emerged as powerful priors for inverse problems \cite{daras2024}. 
Diffusion models, as the state-of-the-art in this family, allow sampling from the posterior distribution $p_{\sv_1,\cdots,\sv_K | \yv}$, enabling approximation of the MMSE estimator $\CondExpt{\sv_1,\cdots,\sv_K}{\yv}$ by posterior sample average and facilitating uncertainty quantification. However, most existing diffusion-model-based approaches for signal separation either train diffusion models to directly sample from the posterior distribution \cite{scheibler2023,dong2025,karchkhadze2025}, requiring retraining when the number or prior distributions of sources change, or lack theoretical guarantees even if the training of diffusion models is perfect \cite{mariani2024}.  

In this work, we propose a posterior sampling algorithm for Bayesian signal separation, based on Gibbs sampling and plug-and-play (PnP) diffusion priors. Our main contributions are threefold:
\begin{enumerate}[label=(\alph*)]
	\item We introduce a modular posterior sampling framework that alternates Gibbs updates of individual sources using their corresponding diffusion models, allowing source priors to be learned separately and flexibly combined without retraining.
	 
	\item We establish a consistency guarantee for the proposed algorithm, showing that it produces exact posterior samples under perfect diffusion model training.  
	
	\item We test the proposed algorithm on the task of extracting heartbeat signals from mixtures with synthetic motion artifacts, demonstrating its superior performance over existing approaches.
\end{enumerate}
\textbf{Notation:} For $i \le j$, let $\sv_{i:j}\triangleq (\sv_i, \sv_{i+1}, \dots, \sv_j)$ in (\ref{eq:observation}). For a random variable $\xv$, we denote a realization of it by $\hat{\xv}$.

\section{Background on Diffusion Models}
\label{sec:DM}

Stochastic differential equation (SDE)-based diffusion models \cite{song2021} have been widely used for modeling complex probability distributions. Let $p_{\text{data}}$ denote the data distribution. By gradually injecting noise through a diffusion process, $p_{\text{data}}$ can be transformed into a Gaussian noise distribution. A diffusion model then performs sampling by simulating the corresponding reverse-time diffusion process, which converts Gaussian noise samples back into samples from $p_{\text{data}}$.

\subsection{Forward and Reverse-Time Diffusion Processes}
\label{sec:forward_reverse_SDEs}

The forward diffusion process is a Markov process $\paren{\xv_t}_{t=0}^T$ with initial state $\xv_0 \sim p_{\text{data}}$ and transition kernel given by the following SDE\footnote{In the general case \cite{song2021}, the forward process can be written as $d\xv_t= f(\xv,t)dt + g(t)d\wv_t$. Here we set $f(\xv,t)\equiv \zerov$ for simplicity, which already covers many state-of-the-art diffusion models \cite{karras2022}.} \cite{song2021,karras2022}:
\begin{equation}\label{eq:forward_SDE}
	d\xv_t = g(t) d\wv_t,
\end{equation}
where $\wv_t$ is a Wiener process and $g(t)>0$ controls the noise injection rate at time $t$. Integrating (\ref{eq:forward_SDE}) yields
\begin{equation}\label{eq:marginal_pdf}
	\xv_t = \xv_0 + \int_{0}^t g(s) d\wv_s \sim p_{ \xv_0 + \sigma(t) \nv} =: p_{\sigma(t)},
\end{equation}
where $\nv\sim\Normal{\zerov}{\Imat}$ is independent of $\xv_0$, and $\sigma(t)\triangleq \sqrt{ \int_{0}^t g^2(s)ds }$ denotes the standard deviation of the injected noise. Notice that the marginal distribution of $\xv_t$ only depends on $\sigma(t)$. When $\sigma(T)$ is much larger than the standard deviation of $p_{\text{data}}$, we may approximate $p_{\xv_T}$ by a Gaussian distribution $\Normal{\zerov}{\sigma^2(T)\Imat}$. Thus as $\xv_t$ evolves from $t=0$ to $T$, the marginal distribution of $\xv_t$ transforms from $p_{\xv_0}=p_{\text{data}}$ to a Gaussian distribution $\Normal{\zerov}{\sigma^2(T)\Imat}$.

The forward process (\ref{eq:forward_SDE}) admits a reverse-time counterpart \cite{anderson1982}. The reverse process is a Markov process $\paren{\bar{\xv}_t}_{t=0}^T$ independent of $\paren{\xv_t}_{t=0}^T$ with initial state being $\bar{\xv}_T\deq\xv_T\sim p_{\sigma(T)}$ and transition kernel given by the reverse-time SDE
\begin{equation}\label{eq:reverse_SDE}
	d\bar{\xv}_t = { - g^2(t) \nabla_{\bar{\xv}_t} \log p_{\sigma(t)}(\bar{\xv}_t) }dt + g(t) d\bar{\wv}_t, 
\end{equation}
where $\bar{\wv}_t$ is a Wiener process independent of $\wv_t$. For any positive integer $N$ and any finite set of times $0\leq t_1< t_2 < \cdots < t_N \leq T$, it is known that \cite{anderson1982}
\begin{equation}\label{eq:Anderson}
	\paren{\bar{\xv}_{t_1},\bar{\xv}_{t_2},\cdots,\bar{\xv}_{t_N}}\deq \paren{{\xv}_{t_1},{\xv}_{t_2},\cdots,{\xv}_{t_N}},
\end{equation}
i.e., if one observes $\bar{\xv}_t$ at a sequence of time points, the joint statistics are indistinguishable from those of $\xv_t$. Thus running $\bar{\xv}_t$ backward in time faithfully reproduces the law of $\xv_t$, which explains why it is referred to as the time reversal of $\xv_t$. 

In particular, as $\bar{\xv}_t$ evolves from $t=T$ to $0$, according to (\ref{eq:marginal_pdf}) and (\ref{eq:Anderson}), the marginal distribution of $\bar{\xv}_t$ transforms from $p_{\bar{\xv}_T}=p_{\xv_T} \approx \Normal{\zerov}{\sigma^2(T)\Imat}$ to $p_{\bar{\xv}_0}=p_{\xv_0}=p_{\text{data}}$. Notice that sampling from $p_{\bar{\xv}_T} \approx \Normal{\zerov}{\sigma^2(T)\Imat}$ is straightforward, simulating \eqref{eq:reverse_SDE} from $t=T$ to $0$ yields samples from $p_{\text{data}}$.

\subsection{Implementation of Diffusion Models}
\label{sec:implementation_dm}

In (\ref{eq:reverse_SDE}), the function $\nabla \log p_{\sigma(t)}(\cdot)$ (termed \textit{score function}) has no closed-form expression and needs to be approximated by a neural network. By Tweedie's formula \cite{efron2011}, for any $\sigma>0$,
\begin{equation}\label{eq:Tweedie}
	(\forall \zv \in \R^d) \;\nabla_{\zv} \log p_\sigma(\zv) = \frac{ \CondExpt{ \xv_0 }{ \xv_0 + \sigma\nv =\zv } - \zv } { \sigma^2 },
\end{equation}
where $\nv\sim\Normal{\zerov}{\Imat}$. Notice that $\CondExpt{ \xv_0 }{ \xv_0 + \sigma\nv =\zv }$ is the MMSE estimator of $\xv_0$ given the noisy observation $\xv_0+\sigma\nv$, thus it can be approximated by a denoising neural network $\CondExpt{ \xv_0 }{ \xv_0 + \sigma\nv =\zv }\approx D_\theta(\zv;\sigma)$, leading to
\begin{equation}\label{eq:score_approx}
	\nabla_{\zv} \log p_\sigma (\zv) \approx \frac{ D_{\theta}(\zv; \sigma) - \zv } { \sigma^2 }.
\end{equation}
The reverse-time diffusion process (\ref{eq:reverse_SDE}) can then be simulated by numerical integration. A common choice is the Euler–Maruyama method, which discretizes $[0,T]$ into steps $T=t_0>t_1>\cdots>t_{M}=0$. Starting from $\hat{\xv}_{t_0}\sim\Normal{\zerov}{\sigma^2(T)\Imat}$, we can simulate (\ref{eq:reverse_SDE}) by updating
\begin{align*}
	\hat{\xv}_{t_{i+1}} \gets \hat{\xv}_{t_i} + g^2(t_i) \frac{ D_{\theta}(\hat{\xv}_{t_i}; \sigma(t_i)) - \hat{\xv}_{t_i} } { \sigma^2(t_i) } h_i + g(t_i)\sqrt{h_i}\,\epsv_i, 
\end{align*}
where $h_i\triangleq t_i-t_{i+1}$, $\epsv_i\sim\Normal{\zerov}{\Imat}$, the score function in (\ref{eq:reverse_SDE}) is replaced by (\ref{eq:score_approx}). Then $\hat{\xv}_{t_M}$ serves as an approximate sample from $p_{\text{data}}$. In practice, the denoising network and numerical solver together constitute a diffusion model.

\section{Method and Analysis}

\begin{algorithm}[t]
	\caption{Gibbs sampling method}
	\label{alg:Gibbs}
	\begin{algorithmic}
		\REQUIRE initial value $\paren{\sv^{(0)}_2, \sv^{(0)}_3, \cdots, \sv^{(0)}_K}$, observation $\hat{\yv}$
		\PARAMETER number of iterations $N$
		\ENSURE posterior sample $\paren{\sv^{(N)}_1, \sv^{(N)}_2, \cdots, \sv^{(N)}_K}$
		\FOR{$i=1$ to $N$}
		\FOR{$k=1$ to $K$}
		\STATE draw $\sv^{(i)}_k  \sim {p}_{\sv_k | \yv , \sv_{1:k-1}, \sv_{k+1:K} } \cparen{\;\cdot\;} { \hat{\yv}, \sv_{1:k-1}^{(i)}, \sv_{k+1:K}^{(i-1)} } $
		\ENDFOR
		\ENDFOR
	\end{algorithmic}
\end{algorithm}

Given observation $\yv$ in (\ref{eq:observation}), we aim to generate samples from the posterior distribution $p_{\sv_{1:K} | \yv}$, which can be further used to approximate the MMSE estimator or other statistics of interest. A standard approach is Gibbs sampling (Algorithm \ref{alg:Gibbs}), which iteratively updates one component of $\sv_{1:K}$ at a time by sampling from its conditional distribution given the others. However, for general prior distributions $p_{\sv_k}$ ($1\leq k\leq K$), the required conditional samplers are intractable, posing a key challenge for practical implementation. 

In this section, we show that with a diffusion model for sampling each prior $p_{\sv_k}$, the conditional sampling in Algorithm \ref{alg:Gibbs} can be realized by partially simulating the reverse diffusion process (\ref{eq:reverse_SDE}). This yields our Diffusion-within-Gibbs (DiG) algorithm for Bayesian signal separation. Next we detail the derivation of the DiG algorithm, and prove that, under perfect training, DiG produces exact posterior samples.

\subsection{An Alternative Formulation of the Conditional PDF}

We first provide an alternative formulation for the conditional pdf $p_{\sv_k \mid \yv, \sv_{1:k-1}, \sv_{k+1:K}}$ appearing in Algorithm \ref{alg:Gibbs}, which will be useful in the following derivation. By (\ref{eq:observation}) and Bayes' theorem, we have
\begin{align}
	& \quad p_{\sv_k\vert \yv, \sv_{1:k-1}, \sv_{k+1:K} }\cparen{ \;\cdot\; } { \hat{\yv}, \sv_{1:k-1}^{(i)}, \sv_{k+1:K}^{(i-1)} } \nonumber\\
	\propto & \quad p_{\sv_k\vert \sv_{1:k-1}, \sv_{k+1:K} }\cparen{ \;\cdot\; } { \sv_{1:k-1}^{(i)}, \sv_{k+1:K}^{(i-1)} }   \nonumber\\
	& \quad \times p_{\yv | \sv_k, \sv_{1:k-1}, \sv_{k+1:K} } \cparen{\hat{\yv}} { \;\cdot\;,  \sv_{1:k-1}^{(i)}, \sv_{k+1:K}^{(i-1)} } \nonumber \\
	\propto & \quad p_{\sv_{k}} ( \cdot ) \; e^{ -\inv{ 2 \sigma_v^2 } \norm{ \rv^{(i)}_k - \;\cdot\;  }^2_2 },  \label{eq:cond_pdf_exp1}
\end{align}
where we define $\rv^{(i)}_k \triangleq \hat{\yv}-\sum_{j=1}^{k-1} \sv^{(i)}_{j} - \sum_{j=k+1}^K \sv^{(i-1)}_{j}$ for brevity; (\ref{eq:cond_pdf_exp1}) follows from independence of sources and (\ref{eq:observation}).

On the other hand, let $\nv \sim \Normal{\zerov}{\Imat}$ be independent of $\sv_k$, and consider the following auxiliary pdf:
\begin{align}
	p_{\sv_k | \sv_k + \sigma_v\nv } \cparen{ \;\cdot\; } { {\rv}^{(i)}_k } \propto & \; p_{\sv_{k}} \paren{ \cdot } \; p_{\sv_k + \sigma_v \nv | \sv_k} \cparen{ {\rv}^{(i)}_k } { \;\cdot\; }  \nonumber \\
	\propto & \;  p_{\sv_{k}} \paren{ \cdot } \; e^{ -\inv{ 2 \sigma_v^2 } \norm{ {\rv}^{(i)}_k - \;\cdot\; }^2_2  } .  \label{eq:cond_pdf_exp2}
\end{align}
Then by comparing (\ref{eq:cond_pdf_exp1}) and (\ref{eq:cond_pdf_exp2}), we can derive
\begin{align}\label{eq:cond_pdf_eq_var}
	\begin{split}
		& \; p_{\sv_k\vert \yv, \sv_{1:k-1}, \sv_{k+1:K} }\cparen{ \;\cdot\; } { \hat{\yv}, \sv_{1:k-1}^{(i)}, \sv_{k+1:K}^{(i-1)} } \\
		= & \; p_{\sv_k | \sv_k + \sigma_v\nv } \cparen{ \;\cdot\; } { \hat{\yv}-\sum\nolimits_{j=1}^{k-1} \sv^{(i)}_{j} - \sum\nolimits_{j=k+1}^K \sv^{(i-1)}_{j} }. 
	\end{split}
\end{align}
Hence to implement Algorithm \ref{alg:Gibbs}, it suffices to sample from $p_{\sv_k \mid \sv_k + \sigma_v \nv}$. Next we show this conditional sampling is equivalent to simulating a partial reverse-time diffusion process.

\subsection{Conditional Sampling Through Diffusion}

For each $1 \leq k \leq K$, suppose that we have access to a dataset of $\sv_k$. Then we can generate new samples from $p_{\sv_k}$ by simulating a reverse-time diffusion process $\paren{\bar{\xv}_{k,t}}_{t=0}^T$, as introduced in Sec.~\ref{sec:DM}. Its initial state is $\bar{\xv}_{k,T} \sim p^{(k)}_{\sigma(T)} \triangleq p_{\sv_k + \sigma(T)\nv}$ and the transition kernel is governed by the reverse-time SDE
\begin{equation}\label{eq:reverse_SDE_sk}
	d\bar{\xv}_{k,t} = - g^2(t)\,\nabla_{\bar{\xv}_{k,t}} \log p^{(k)}_{\sigma(t)}\paren{\bar{\xv}_{k,t}} \, dt + g(t) \, d\bar{\wv}_{k,t},
\end{equation}
where $ p^{(k)}_{\sigma} \triangleq p_{\sv_k + \sigma\nv}$. Recalling (\ref{eq:Anderson}) and (\ref{eq:marginal_pdf}), one can verify that
\begin{align}
	p_{\bar{\xv}_{k,0} | \bar{\xv}_{k,t}}\cparen{ \hat{\xv}_{k,0} } { \hat{\xv}_{k,t} } &= p_{{\xv}_{k,0} | {\xv}_{k,t}}\cparen{ \hat{\xv}_{k,0} } { \hat{\xv}_{k,t} } \nonumber \\
	&= p_{\sv_k | \sv_k + \sigma(t)\nv} \cparen{ \hat{\xv}_{k,0} } { \hat{\xv}_{k,t} } \label{eq:diff_eq_cond_pdf}.
\end{align}
Then by (\ref{eq:cond_pdf_eq_var}), if there exists $t_v\in[0,T]$ satisfying $\sigma(t_v)=\sigma_v$, the following holds:
\begin{align}\label{eq:diff_equiv_cond_sample}
	& \; p_{\sv_k\vert \yv, \sv_{1:k-1}, \sv_{k+1:K} }\cparen{ \;\cdot\; } { \hat{\yv}, \sv_{1:k-1}^{(i)}, \sv_{k+1:K}^{(i-1)} } \nonumber\\
	= & \; p_{\bar{\xv}_{k,0} \mid \bar{\xv}_{k,t_v}} \cparen{ \;\cdot\; } { \hat{\yv}-\sum\nolimits_{j=1}^{k-1} \sv^{(i)}_{j} - \sum\nolimits_{j=k+1}^K \sv^{(i-1)}_{j} }. \nonumber
\end{align}
Thus by setting $\bar{\xv}_{k,t_v} \gets \hat{\yv}-\sum\nolimits_{j=1}^{k-1} \sv^{(i)}_{j} - \sum\nolimits_{j=k+1}^K \sv^{(i-1)}_{j}$ and simulating (\ref{eq:reverse_SDE_sk}) from $t=t_v$ to $0$, the resultant $\bar{\xv}_{k,0}$ can be regarded as a realization of $\sv^{(i)}_k$ in Algorithm \ref{alg:Gibbs}. Applying this result to Algorithm \ref{alg:Gibbs} yields the proposed \textit{Diffusion-within-Gibbs} (DiG) algorithm; see Algorithm~\ref{alg:DiG}. 

Remarkably, the DiG algorithm only requires pre-trained diffusion models for the priors of individual source signals. Hence, unlike existing methods \cite{scheibler2023,dong2025,karchkhadze2025} that require retraining when the number of sources or the prior distributions change, DiG can flexibly plug pre-trained diffusion models into various signal separation problems.

\begin{algorithm}[t]
	\caption{Diffusion-within-Gibbs (DiG) algorithm}
	\label{alg:DiG}
	\begin{algorithmic}
		\REQUIRE initial value $\paren{\sv^{(0)}_2, \sv^{(0)}_3, \cdots, \sv^{(0)}_K}$, observation $\hat{\yv}$
		\PARAMETER number of iterations $N$, noise level $\sigma_v$
		\PREPARATION train $K$ diffusion models for simulating the reverse-time SDE (\ref{eq:reverse_SDE_sk}) for all $k\in \bracset{1,2,\dots,K}$
		\ENSURE posterior sample $\paren{\sv^{(N)}_1, \sv^{(N)}_2, \cdots, \sv^{(N)}_K}$
		
		\STATE find $t_v\in[0,T]$ satisfying $\sigma(t_v)=\sigma_v$
		\FOR{$i=1$ to $N$}
		\FOR{$k=1$ to $K$}
		\STATE $\bar{\xv}_{k,t_v}\gets \hat{\yv}-\sum_{j=1}^{k-1} \sv^{(i)}_{j} - \sum_{j=k+1}^K \sv^{(i-1)}_{j}$
		\STATE obtain $\bar{\xv}_{k,0}$ by simulating (\ref{eq:reverse_SDE_sk}) from $t=t_v$ to $0$
		\STATE $\sv^{(i)}_k\gets \bar{\xv}_{k,0}$
		\ENDFOR
		\ENDFOR
	\end{algorithmic}
\end{algorithm}

\subsection{Consistency with an Ideal Diffusion Simulator}

In the following, we establish the consistency of DiG, namely that under an ideal diffusion simulator, it produces exact posterior samples.

\begin{theorem}\label{thm:Gibbs_convergence}
	Assume that for all $1\leq k\leq K$, $\sv_k$ has a strictly positive density with respect to the Lebesgue measure, and that the simulation of (\ref{eq:reverse_SDE_sk}) is exact.\footnote{Here, ``exact simulation of (\ref{eq:reverse_SDE_sk})" means that, in the diffusion model used to simulate (\ref{eq:reverse_SDE_sk}), the denoising network perfectly learns the MMSE estimator and the discretization step size in the numerical solver vanishes.} Then the joint distribution of $\sv^{(N)}_{1:K}$ generated by Algorithm \ref{alg:DiG} converges in total variation distance to the true posterior $p_{\sv_{1:K} \mid \yv=\hat{\yv} }$ as $N \to \infty$.
\end{theorem}
\begin{proof}
	When (\ref{eq:reverse_SDE_sk}) is simulated exactly, the $i$th iteration of Algorithm \ref{alg:DiG} is equivalent to that of Algorithm \ref{alg:Gibbs}. Hence, the Markov transition kernel from $\sv^{(i-1)}_{1:K}$ to $\sv^{(i)}_{1:K}$, denoted by $M$, is reversible with respect to $p_{\sv_{1:K} \mid \yv=\hat{\yv} }$. Since for each $1\leq k\leq K$, $p_{\sv_k}$ is strictly positive everywhere, one can verify that $M$ is irreducible and aperiodic. Thus we can derive the result from \cite[Thm.~1 and Cor.~1]{tierney1994}.
\end{proof}

Recently, several other diffusion-model-based posterior sampling methods have also been shown to achieve consistency with ideal diffusion simulators for more general inverse problems \cite{daras2024}. However, such guarantees are typically asymptotic, requiring either the number of samples maintained in each iteration to approach infinity \cite{dou2024} or certain annealing parameter to vanish arbitrarily slowly \cite{xu2024}, thus providing weaker guarantees than the proposed DiG algorithm.

\section{Experiments}

\begin{table}
	\centering
	\setlength{\tabcolsep}{3pt} % 调小列间距
	\caption{Mean squared error of recovered heartbeat signals.}
	\resizebox{\linewidth}{!}{ % 缩放到页面宽度
		\begin{tabular}{c|ccccc}
			\hline
			(SIR, SNR) (dB) & EMD & VMD & MSDM & DPnP & Ours \\
			\hline
			(-20.1, 13.2) &  9.17 & 0.33 & 0.20 & 0.32 & \textbf{0.19} \\
			(-26.1, 13.2) & 18.34 & 0.62 & 0.37 & 0.44 & \textbf{0.26} \\
			(-40.1, 13.2) & 90.24 & 27.06 & 4.06 & 0.98 & \textbf{0.57} \\
			(-20.1, -0.8) &  9.21 & 0.79 & 0.85 & 0.41 & \textbf{0.28} \\
			(-26.1, -0.8) & 18.36 & 0.84 & 0.92 & 0.50 & \textbf{0.31} \\
			(-40.1, -0.8) & 91.62 & 20.86 & 4.23 & 0.93 & \textbf{0.61} \\
			(-20.1, -6.8) &  9.30 & 1.44 & 1.90 & 0.59 & \textbf{0.37} \\
			(-26.1, -6.8) & 18.41 & 1.42 & 1.95 & 0.65 & \textbf{0.39} \\
			(-40.1, -6.8) & 91.64 & 5.23 & 4.80 & 0.93 & \textbf{0.68} \\
			\hline
		\end{tabular}
	}
	\label{tab1}
\end{table}

We evaluate the DiG algorithm on the task of extracting heartbeat signals from mixtures with synthetic motion artifacts, motivated by mmFMCW radar applications for heartbeat monitoring under large-scale body motion \cite{han2023cough}. Although this task involves only two sources, the large amplitude disparity between the desired heartbeat and the motion interference presents a significant challenge.

\textbf{Datasets}: Heartbeat training signals are drawn from the impedance dataset \cite{schellenberger2020dataset}, bandpass-filtered, and segmented into 10-s clips from 25 subjects, yielding 50,000 samples. Motion signals are generated from 10-s velocity profiles with randomized piecewise-constant amplitudes and sigmoidal transitions. Both components are normalized before training. For evaluation, measurements are formed by summing new heartbeat segments (signal of interest) from 5 held-out subjects with motion segments (interference) and additive Gaussian noise (observational noise), producing multiple test sets that cover different combinations of signal-to-interference ratio (SIR) and signal-to-noise (SNR) ratio.

\textbf{Diffusion Models}: We train two diffusion models respectively for heartbeat and motion signal generation. Following \cite{song2021}, we set the noise injection rate $g(t)\triangleq \alpha^t$ in (\ref{eq:reverse_SDE_sk}) with $\alpha=15$. Both diffusion models employ denoising networks with a WaveNet-inspired \cite{van2016wavenet} design.

\textbf{Algorithms for comparison}: We compare the proposed DiG algorithm with four methods: EMD \cite{huang1998}, VMD \cite{dragomiretskiy2014}, MSDM \cite{mariani2024} and DPnP \cite{xu2024}. The first two are classical model-driven decomposition methods, while the latter two are state-of-the-art posterior sampling methods using PnP diffusion priors. For EMD and VMD, which only decompose the observation into multiple periodic modes, we select a subset of modes whose sum best approximates the heartbeat signal. Specifically, we choose the combination of modes that minimizes the $\ell_2$-distance to the ground truth.\footnote{This oracle selection assumes access to the true signal; in practical scenarios without ground truth, performance is expected to degrade.} In all diffusion-based sampling methods, we draw 25 posterior samples for each observation and compute their average as an estimate of the heartbeat signal.

\textbf{Results}: The evaluation of different methods was performed on a test set of 200 samples under various combinations of SIR and SNR. Table~\ref{tab1} reports the mean squared error (MSE) between the ground truth (GT) and the recovered heartbeat signals. Two representative examples of recovered signals, corresponding to (SIR, SNR) = (-40.1 dB, 13.2 dB) and (-26.1 dB, -0.8 dB), are illustrated in Fig.~\ref{fig1}. From Table~\ref{tab1} and Fig.~\ref{fig1}, the proposed DiG algorithm consistently outperforms all comparison methods across the full range of SIR and SNR settings, achieving more accurate heartbeat recovery even under strong motion interference or low SNR.
\begin{figure}
	\centering
	\includegraphics[width=.98\linewidth]{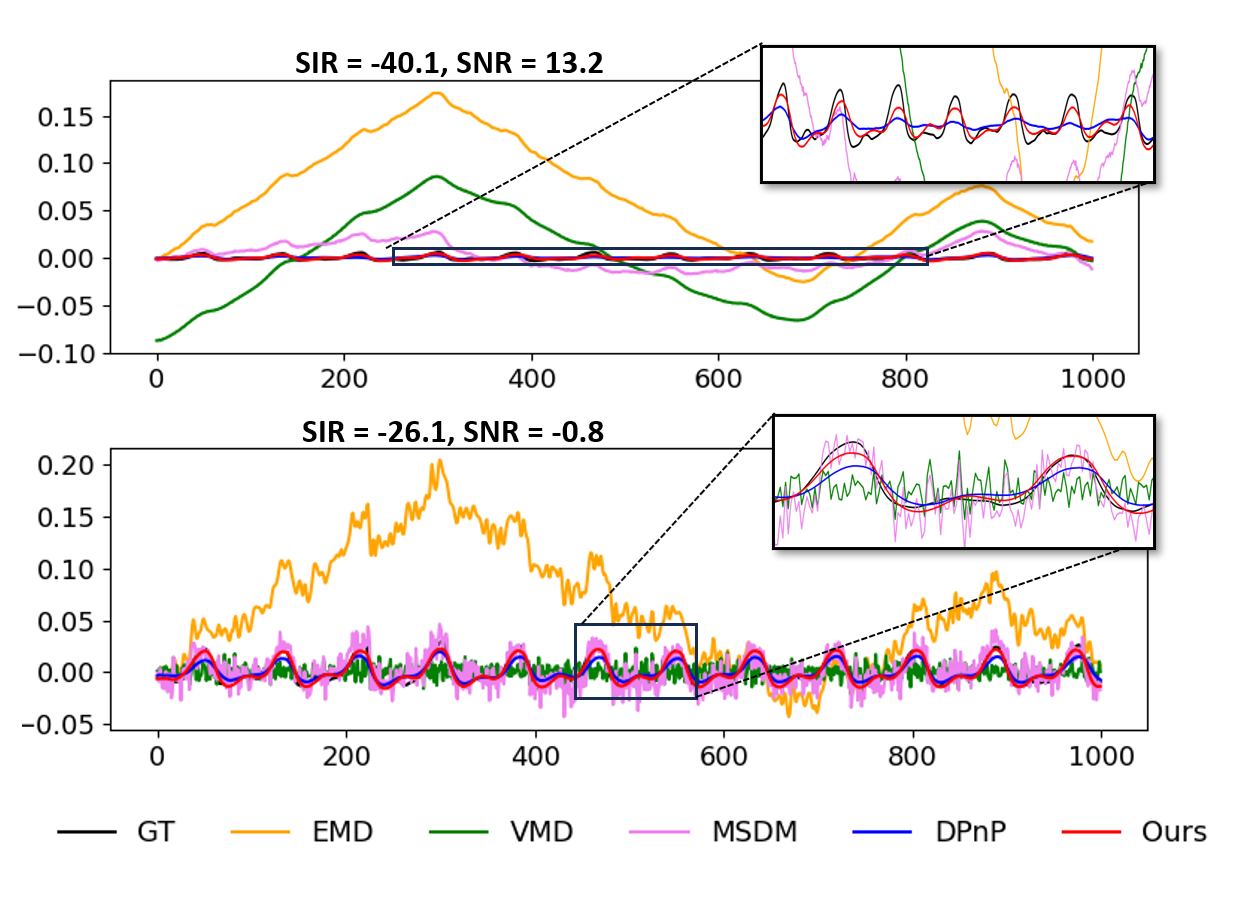} 
	\vspace{-7mm}
	
	\caption{Representative recovered heartbeats with (SIR, SNR) = (-40.1 dB, 13.2 dB) and (-26.1 dB, -0.8 dB), respectively. }
	\label{fig1}
\end{figure}
\vspace{-2mm}
\section{Conclusion}
We presented a posterior sampling algorithm for Bayesian signal separation, which combines Gibbs sampling with plug-and-play diffusion priors. 
Our method allows flexible incorporation of diverse source distributions without retraining and is theoretically consistent, converging to the true posterior using ideal diffusion models. 
Experiments on heartbeat extraction have demonstrated superior performance of the proposed method over existing approaches. 

\vfill\pagebreak

% -------------------------------------------------------------------------
\bibliographystyle{IEEEtran}
\bibliography{refs}

\end{document}